\documentclass[11pt]{article}
\usepackage{amssymb}
\usepackage{amsfonts}
\usepackage{amsmath}

\newtheorem{theorem}{Theorem}

\newtheorem{lemma}{Lemma}

\newtheorem{proposition}{Proposition}

\begin{document}

\title{A Simple Method for School Choice Lotteries\thanks{%
The author is grateful to Minoru Kitahara for insightful comments and
suggestions. This work was supported by JSPS KAKENHI Grant Number 25K05004.}}
\author{Yasunori Okumura\thanks{%
Department of Logistics and Information Engineering, Tokyo University of
Marine Science and Technology (TUMSAT), 2-1-6 Etchujima, Koto-ku, Tokyo
135-8533, Japan, Phone: +81-3-5245-7300, Fax: +81-3-5245-7300, E-mail:
okuyasu@gs.econ.keio.ac.jp}}
\maketitle

\begin{center}
\textbf{Abstract}
\end{center}

This note proposes a simple polynomial-time method for constructing an ex
ante stable school-choice lottery satisfying equal treatment of equals
(ETE). We show that the ETE reassignment of any constrained efficient stable
matching is ex ante stable, satisfies ETE, and is not ordinally dominated by
any other ex ante stable lottery. We further show that there exists a
constrained efficient stable matching whose ETE reassignment is not
ordinally dominated by any ex post stable lottery.

\textbf{JEL classification}: C78, D63, D47

\textbf{Keywords: }School choice lottery, Ex ante stability, Equal treatment
of equals, Efficiency \newpage

\section{Introduction}

Much of the literature on school choice has focused on deterministic
matchings alone. In real-world school choice systems, however, ties are
sometimes broken by lottery, so that the resulting assignment should be
represented as a probability distribution over deterministic matchings. From
this perspective, Kesten and \"{U}nver (2015) model a school choice
mechanism as a lottery mechanism and analyze it using the notion of ex ante
stability.

In particular, Kesten and \"{U}nver (2015) introduce two mechanisms. First,
FDA yields an efficient lottery subject to ex ante stability and the
requirement that students with the same priority be treated equally in each
school. Second, FDAT relaxes this requirement by permitting differential
treatment among students with the same priority when they have different
preferences, and yields a lottery that is not ordinally dominated by any
other ex ante stable lottery. In other words, FDAT yields a lottery that is
constrained efficient within the class of ex ante stable ones and satisfies
equal treatment of equals.

However, the computational tractability of these mechanisms is problematic.
Indeed, Cookson and Shah (2025) present a counterexample on which the FDA
enters an infinite loop and hence does not terminate in finite time.%
\footnote{%
Cookson and Shah (2025) propose a polynomial-time algorithm, DFDA-SCC, but
note that it is not known whether this algorithm always produces the same
outcome as FDA.} The same difficulty arises for FDAT as well, since FDAT
uses FDA in its initial step to obtain the starting lottery.

In light of this issue, this note proposes, as an alternative to FDAT, a
very simple polynomial-time method that yields an ex ante stable lottery
satisfying equal treatment of equals that is not ordinally dominated by any
other ex ante stable lottery. The method proceeds as follows. First, we
compute a deterministic constrained efficient matching; that is, a stable
matching that is not Pareto dominated by any other stable matching. Second,
we derive the ETE reassignment, which is the reassignment method introduced
by Okumura (2025), of the matching obtained in the first step.

Moreover, we consider a weaker stability concept of ex post stability. Since
ex ante stability implies ex post stability but not conversely, the ETE
reassignment of a deterministic lottery induced by a constrained efficient
matching may be ordinally dominated by an ex post stable lottery. However,
we show that there exists some stable matching such that the ETE
reassignment of the stable matching is not ordinally dominated by any ex
post stable matching.

\section{Model}

We introduce the school choice problem as considered by Kesten and \"{U}nver
(2015). A school choice problem is a five-tuple 
\begin{equation*}
\left( I,C,q,P,\succsim \right) ,
\end{equation*}%
where $I$ and $C$ are finite sets of students and schools, respectively. The
vector $q=\left( q_{c}\right) _{c\in C}$ is a quota profile, where $q_{c}\in 
\mathbb{N}$ represents the simple quota constraint of $c$. Here, we assume $%
\sum\nolimits_{c\in C}q_{c}\geq \left\vert I\right\vert $.

The preference profile $P=\left( P_{i}\right) _{i\in I}$ specifies students'
preferences, where $P_{i}$ is a strict linear order over $C$ (satisfying
transitivity, totality and irreflexivity) representing the strict preference
ranking of the schools of student $i$. If $i$ prefers school $c$ to school $%
c^{\prime }$, we write $cP_{i}c^{\prime }$. Moreover, we write $%
cR_{i}c^{\prime }$ if $cP_{i}c^{\prime }$ or $c=c^{\prime }$.

Finally, $\succsim =\left( \succsim _{c}\right) _{c\in C}$ is a priority
profile, where $\succsim _{c}$ is a weak order over $I$ (satisfying
transitivity and completeness), representing the weak ranking of the
students of school $c$. If $i\succsim _{c}j$ and $\lnot \left( j\succsim
_{c}i\right) $, we write $i\succ _{c}j$ meaning that student $i$ has a
higher priority than student $j$ for school $c$. On the other hand, if $%
i\succsim _{c}j$ and $j\succsim _{c}i$, we write $i\sim _{c}j$ meaning that $%
i$ and $j$ are tied for school $c$.

Let $\mu $ be a (deterministic)\textbf{\ matching} satisfying for all $i\in
I $ and all $c\in C$, $\mu \left( i\right) \in C$, $\mu \left( c\right)
\subseteq I,$ $\mu \left( i\right) =c$ if and only if $i\in \mu \left(
c\right) ,$ and $\left\vert \mu \left( c\right) \right\vert \leq q_{c}$. Let 
$\mathcal{M}$ be the set of all possible matchings.

A matching $\mu $ is said to have an (ex post) \textbf{justified envy} if
there are two students $i$ and $j$ such that $\mu \left( j\right) P_{i}\mu
\left( i\right) $ and $i\succ _{\mu \left( j\right) }j$.\ A matching $\mu $
is said to be \textbf{stable} if $\mu $ has no justified envy. A matching $%
\mu $ is said to be \textbf{Pareto dominated} by another matching $\mu
^{\prime }$ if $\mu ^{\prime }(i)R_{i}\mu (i)$ for all $i\in I$ and $\mu
^{\prime }(j)P_{j}\mu (j)$ for some $j\in I$. A matching $\mu $ is said to
be \textbf{constrained efficient }if $\mu $ is stable and not Pareto
dominated by any other stable matching.

Let $\lambda =\left( \lambda _{\mu }\right) _{\mu \in \mathcal{M}}$ be a 
\textbf{lottery} satisfying $\lambda _{\mu }\in \left[ 0,1\right] $ for all $%
\mu \in \mathcal{M}$ and $\sum\nolimits_{\mu \in \mathcal{M}}\lambda _{\mu
}=1$. That is, a lottery is a probability distribution over matchings. Let $%
\Delta \mathcal{M}$ be the set of all possible lotteries. Moreover, let 
\begin{equation*}
\mathcal{M}\left( \lambda \right) =\left\{ \left. \mu \in \mathcal{M}\text{ }%
\right\vert \text{ }\lambda _{\mu }>0\right\} 
\end{equation*}%
be the support of lottery $\lambda $. Specifically, let $\lambda ^{\mu }\in
\Delta \mathcal{M}$ be the \textbf{deterministic} \textbf{lottery} that
assigns probability one to $\mu $, that is, $\lambda _{\mu }^{\mu }=1$.

Kesten and \"{U}nver (2015) define a random matching as a stochastic matrix
of assignment probabilities. A lottery over deterministic matchings induces
such a random matching, and they study stability and related properties at
the level of the induced random matching. By contrast, we analyze the
properties of the lottery itself. A similar lottery-based approach is also
adopted by Kesten et al. (2017).

For each $i\in I$ and $\lambda \in \Delta \mathcal{M}$, let $\mathbf{c}%
\left( i,\lambda \right) $ be a random school (variable) assigned to $i$.
Formally, for each $c\in C$, $\mathbf{c}\left( i,\lambda \right) =c$ with
probability%
\begin{equation*}
\Pr (c,\mathbf{c}\left( i,\lambda \right) )=\sum\limits_{\mu :\text{ }\mu
\left( i\right) =c}\lambda _{\mu }\text{.}
\end{equation*}%
Finally, for each $c\in C$, define 
\begin{equation*}
\bar{F}\left( c,i,\lambda \right) =\sum\limits_{c^{\prime }:c^{\prime
}R_{i}c}\Pr (c^{\prime },\mathbf{c}\left( i,\lambda \right) ),
\end{equation*}%
which is the probability that student $i$ is assigned to a school that $i$
weakly prefers to $c$ under $\lambda $.

We say that for $i\in I$, $\mathbf{c}\left( i,\lambda \right) $ is \textbf{%
weakly} (first-order) \textbf{stochastically dominated} by $\mathbf{c}\left(
i,\lambda ^{\prime }\right) $ if $\bar{F}\left( c,i,\lambda ^{\prime
}\right) \geq \bar{F}\left( c,i,\lambda \right) $ for all $c\in C$. On the
other hand, $\mathbf{c}\left( i,\lambda \right) $ is \textbf{strictly
stochastically dominated} by $\mathbf{c}\left( i,\lambda ^{\prime }\right) $
if $\bar{F}\left( c,i,\lambda ^{\prime }\right) \geq \bar{F}\left(
c,i,\lambda \right) $ for all $c\in C$ and $\bar{F}\left( c,i,\lambda
^{\prime }\right) >\bar{F}\left( c,i,\lambda \right) $ for some $c\in C$. A
lottery $\lambda \in \Delta \mathcal{M}$ is \textbf{ordinally dominated} by
another lottery $\lambda ^{\prime }\in \Delta \mathcal{M}$ if for all $i$, $%
\mathbf{c}\left( i,\lambda \right) $ is weakly stochastically dominated by $%
\mathbf{c}\left( i,\lambda ^{\prime }\right) ,$ and for some $i$, $\mathbf{c}%
\left( i,\lambda \right) $ is strictly stochastically dominated by $\mathbf{c%
}\left( i,\lambda ^{\prime }\right) $.

We provide two stability notions on lotteries.

A lottery $\lambda $ is said to have \textbf{ex ante justified envy} if
there are two students $i$ and $j$, school $s$ and matchings $\mu ,\mu
^{\prime }\in \mathcal{M}\left( \lambda \right) ,$ such that $sP_{i}\mu
\left( i\right) $, $i\succ _{s}j,$ and $\mu ^{\prime }\left( j\right) =s$.
Note that if $\mu =\mu ^{\prime };$ that is, if $\mu \left( j\right)
=sP_{i}\mu \left( i\right) $ and $i\succ _{s}j$, then $\mu $ has a justified
envy. A lottery $\lambda $ is said to be \textbf{ex ante stable} if it has
no ex ante justified envy.

We immediately have the following result.

\begin{lemma}
If $\mu $ is stable, then $\lambda ^{\mu }$ is ex ante stable. Moreover, if $%
\mu $ is constrained efficient, then $\lambda ^{\mu }$ is not ordinally
dominated by any ex ante stable lottery.
\end{lemma}

Let $I_{1},\cdots ,I_{N}$ be a partition of $I$ such that $I_{1}\cup \cdots
\cup I_{N}=I$ and $I_{n}\cap I_{m}=\emptyset $ for any $n,m=1,\cdots ,N$
where $n\neq m,$ and for any $n=1,\cdots ,N$ and any $i,j\in I_{n}$, $%
P_{i}=P_{j}$ and $i\sim _{c}j$ for all $c\in C$. That is, for each $%
n=1,\cdots ,N,$ all members of $I_{n}$ are \textit{equals} with respect to
their preferences and the priority orders at every school. Note that we
allow that for $i\in I_{n}$ and $j\in I_{m}$ with $n\neq m,$ $P_{i}=P_{j}$
and $i\sim _{c}j$ for all $c\in C$, for example this may happen when $i$ is
a majority student and $j$ is a minority student. See Okumura (2026) for a
detailed discussion of this point. We refer to each $I_{n}$ as \textbf{a
group of equals}.

A lottery $\lambda \in \Delta \mathcal{M}$ is said to satisfy \textbf{equal
treatment of equals (}hereafter\textbf{\ ETE) }if for each $n=1,\cdots ,N$
and any $i,j\in I_{n}$, $\mathbf{c}\left( i,\lambda \right) =\mathbf{c}%
\left( j,\lambda \right) $. That is, the requirement is that all members of
the same group of equals be assigned the same probability distribution over
schools.

Okumura (2025) provides a simple method for deriving a lottery that
satisfies ETE from a given initial lottery. Let $\pi :I\rightarrow I$ be a
bijection satisfying $\pi \left( i\right) =j$ implies $i,j\in I_{n}$ for $%
n=1,\cdots ,N$. Let $\pi ^{1},\pi ^{2},\cdots ,\pi ^{L}$ be distinct
possible such bijections where 
\begin{equation*}
L=\left\vert I_{1}\right\vert !\times \cdots \times \left\vert
I_{N}\right\vert !.
\end{equation*}%
Moreover, let 
\begin{equation*}
L_{-n}=\left\vert I_{1}\right\vert !\times \cdots \times \left\vert
I_{n-1}\right\vert !\times \left\vert I_{n+1}\right\vert !\times \cdots
\times \left\vert I_{N}\right\vert !.
\end{equation*}%
Fix an arbitrary $\mu \in \mathcal{M}$. We let for $l=1,\cdots ,L$, $\mu
^{l} $ be such that $\mu ^{l}\left( i\right) =\mu \left( \pi ^{l}(i)\right) $
for all $i\in I$.

We let a multiset $\mathcal{M}_{D}\left( \mu \right) =\left\{ \left\{ \mu
^{1},\cdots ,\mu ^{L}\right\} \right\} $ and say that each element of $%
\mathcal{M}_{D}\left( \mu \right) $ is \textbf{derived from} $\mu $.%
\footnote{%
Even if $l\neq l^{\prime },$ $\mu ^{l}=\mu ^{l^{\prime }}$ may be
satisified. Therefore, $Y_{D}\left( y\right) $ is a multiset.} Let $\lambda
\left( \mu \right) \in \Delta \mathcal{M}$ be the lottery such that 
\begin{equation*}
\lambda \left( \mu \right) _{\mu ^{\prime }}=\frac{\left\vert \left\{
l=1,\cdots ,L\text{ }\left\vert \text{ }\mu ^{l}=\mu ^{\prime }\right.
\right\} \right\vert }{L}.
\end{equation*}%
Since $\mu $ is derived from itself, $\lambda \left( \mu \right) _{\mu }\geq
1/L$.

For a given lottery $\lambda $, we say that a lottery $\lambda ^{\prime }$
is the \textbf{ETE reassignment} of $\lambda $ if for all $\mu ^{\prime }\in 
\mathcal{M}$,

\begin{equation*}
\lambda _{\mu ^{\prime }}^{\prime }=\sum\limits_{\mu \in \mathcal{M}}\lambda
_{\mu }\times \lambda \left( \mu \right) _{\mu ^{\prime }}.
\end{equation*}%
To illustrate the ETE reassignment of $\lambda $, consider a simple case in
which $\lambda _{\mu }=1$. Suppose that there are $n\geq 2$ students who
belong to the same equals group. Then, the schools assigned to them under $%
\mu $ may differ across students. In the ETE reassignment of $\lambda $, the
schools assigned to these students under $\mu $ are pooled and then
reassigned uniformly among them.

In fact, Okumura (2025, Lemma 1) shows the following result.

\begin{lemma}
(Okumura, 2025) Let $\lambda ^{\prime }$ be the ETE reassignment of $\lambda 
$. Then, for all $i\in I_{n}$ and all $c\in C$, 
\begin{equation*}
\Pr (c,\mathbf{c}\left( i,\lambda ^{\prime }\right) )=\frac{1}{\left\vert
I_{n}\right\vert }\sum\limits_{j\in I_{n}}\Pr \left( c,\mathbf{c}\left(
j,\lambda \right) \right) .
\end{equation*}
\end{lemma}

This result shows that the ETE reassignment, which may at first appear to be
complicated, can in fact be obtained by the following simple procedure.
First, a matching $\mu \in \mathcal{M}\left( \lambda \right) $ is realized
according to $\lambda $. Second, for each group of equals, the assignments
received under $\mu $ are pooled and then randomly redistributed among the
members of the group with equal probabilities.

Moreover, we have the following result due to Okumura (2025, Proposition 1).

\begin{lemma}
(Okumura, 2025) For any $\lambda \in \Delta \mathcal{M}$, the ETE
reassignment of $\lambda $ satisfies ETE.
\end{lemma}

\section{Ex ante stability}

First, we consider the ETE reassignment of an ex ante stable lottery. We
have the following result.

\begin{proposition}
If $\lambda $ is an ex ante stable lottery, then the ETE reassignment of $%
\lambda $ is also ex ante stable.
\end{proposition}

\textbf{Proof. }Let\textbf{\ }$\lambda ^{\prime }$ be the ETE reassignment
of $\lambda $. Suppose not; that is, $\lambda ^{\prime }$ is not ex ante
stable. Then, there are two students $i$ and $j$, school $s$ and matchings $%
\mu ,\mu ^{\prime }\in \mathcal{M}\left( \lambda ^{\prime }\right) ,$ such
that $sP_{i}\mu \left( i\right) $, $i\succ _{s}j,$ and $\mu ^{\prime }\left(
j\right) =s$. Let $n,m=1,\cdots ,N$ be such that $i\in I_{n}$ and $j\in I_{m}
$. Since $\lambda ^{\prime }$ is the ETE reassignment of $\lambda ,$ $\mu $
and $\mu ^{\prime }$ are derived from matchings $\bar{\mu}$ and $\bar{\mu}%
^{\prime }$ that are included in $\mathcal{M}\left( \lambda \right) ,$
respectively. Then, there exist $i^{\prime }\in I_{n}$ and $j^{\prime }\in
I_{m}$ such that $\bar{\mu}(i^{\prime })=\mu \left( i\right) $ and $\bar{\mu}%
^{\prime }\left( j^{\prime }\right) =s$. Then, since $i^{\prime }\in I_{n}$
and $j^{\prime }\in I_{m}$, $sP_{i^{\prime }}\bar{\mu}(i^{\prime })\,$and $%
i^{\prime }\succ _{s}j^{\prime }$ because%
\begin{equation*}
i^{\prime }\sim _{s}i\succ _{s}j\sim _{s}j^{\prime }.
\end{equation*}%
However, these facts contradict that $\lambda $ has no ex ante justified
envy. Thus, $\lambda ^{\prime }$\textbf{\ }is ex ante stable.\textbf{\ Q.E.D.%
}\newline

Second, we consider the ETE reassignment of a more specific ex ante stable
lottery.

\begin{theorem}
If $\mu ^{\ast }$ is a constrained efficient matching, then the ETE
reassignment of $\lambda ^{\mu ^{\ast }}$ is an ex ante stable lottery that
is not ordinally dominated by any other ex ante stable lottery and satisfies
ETE.
\end{theorem}

\textbf{Proof. }Let $\lambda ^{\ast }$ be the ETE reassignment of $\lambda
^{\mu ^{\ast }}$. First, by Lemma 1 and Proposition 1, $\lambda ^{\ast }$ is
ex ante stable. Second, by Lemma 3, $\lambda ^{\ast }$ must satisfy ETE.

We show that $\lambda ^{\ast }$ is not ordinally dominated by any other ex
ante stable lottery. We use Proposition 5 of Kesten and \"{U}nver (2015).
For a lottery $\lambda $, if there are two students $i$ and $j$ and two
schools $c$ and $d$ such that 
\begin{equation*}
dP_{i}c,\text{ }\Pr (c,\mathbf{c}\left( i,\lambda \right) )>0,\text{ and }%
\Pr (d,\mathbf{c}\left( j,\lambda \right) )>0,
\end{equation*}%
then we write%
\begin{equation*}
\left( i,c\right) \gtrdot ^{\lambda }\left( j,d\right) .
\end{equation*}%
Moreover, for $\lambda $, if there are two students $i$ and $j$ and two
schools $c$ and $d$ such that $\left( i,c\right) \gtrdot ^{\lambda }\left(
j,d\right) $ and $i\succsim _{d}k$ for all $(k,c^{\prime })$ satisfying $%
(k,c^{\prime })\gtrdot ^{\lambda }\left( j,d\right) $, then we write%
\begin{equation*}
\left( i,c\right) \blacktriangleright ^{\lambda }\left( j,d\right) .
\end{equation*}%
Let an \textbf{ex ante stable improvement cycle} $(i_{1},c_{1},\ldots
,i_{M},c_{M})$ at $\lambda $ be a list of distinct student-school pairs such
that 
\begin{equation*}
\left( i_{1},c_{1}\right) \blacktriangleright ^{\lambda }\left(
i_{2},c_{2}\right) \blacktriangleright ^{\lambda }\cdots \blacktriangleright
^{\lambda }\left( i_{M},c_{M}\right) \blacktriangleright ^{\lambda }\left(
i_{1},c_{1}\right) .
\end{equation*}

\begin{lemma}
(Kesten and \"{U}nver 2015, Proposition 5) An ex ante stable lottery $%
\lambda $ is not ordinally dominated by any other ex ante stable lottery if
and only if there exists no ex ante stable improvement cycle at $\lambda $.
\end{lemma}

Now, we show that $\lambda ^{\ast }$ is not ordinally dominated by any other
ex ante stable lottery. Suppose not; that is, $\lambda ^{\ast }$ is
ordinally dominated by some ex ante stable lottery. Then, by Lemma 4, there
exists an ex ante stable improvement cycle at $\lambda ^{\ast }$. Let $%
(i_{1},c_{1},\ldots ,i_{M},c_{M})$ be the cycle. Since $\Pr (c_{m},\mathbf{c}%
\left( i_{m},\lambda ^{\ast }\right) )>0$ for all $m=1,\cdots ,M$, there is $%
\mu _{m}\in \mathcal{M}\left( \lambda ^{\ast }\right) $ such that $\mu
_{m}\left( i_{m}\right) =c_{m}$ for all $m=1,\cdots ,M$. We arbitrary fix $%
m=1,\cdots ,M$ and let $i_{m}\in I_{n}$ for $n=1,\cdots ,N$. Since $\lambda
^{\ast }$ is the ETE reassignment of $\lambda ^{\mu ^{\ast }}$, $\mu _{m}$
is derived from $\mu ^{\ast }$ and therefore there is $i_{m}^{\prime }\in
I_{n}$ such that $\mu ^{\ast }\left( i_{m}^{\prime }\right) =c_{m}$. Thus, 
\begin{equation*}
\left( i_{m}^{\prime },c_{m}\right) \gtrdot ^{\lambda ^{\ast }}\left(
i_{m+1}^{\prime },c_{m+1}\right)
\end{equation*}%
for all $m=1,\cdots ,M,$ where $\left( i_{M+1}^{\prime },c_{M+1}\right)
=\left( i_{1}^{\prime },c_{1}\right) $.

We show that $i_{m}^{\prime }\succsim _{c_{m+1}}j$ for all $(j,c^{\prime })$
satisfying $(j,c^{\prime })\gtrdot ^{\lambda ^{\mu ^{\ast }}}\left(
i_{m+1}^{\prime },c_{m+1}\right) $. Suppose not; that is, there is $%
(j,c^{\prime })$ satisfying $(j,c^{\prime })\gtrdot ^{\lambda ^{\mu ^{\ast
}}}\left( i_{m+1}^{\prime },c_{m+1}\right) \,$and $j\succ
_{c_{m+1}}i_{m}^{\prime }$. Then, $\mu ^{\ast }(j)=c^{\prime }$. Since $%
\lambda ^{\ast }$ is the ETE reassignment of $\lambda ^{\mu ^{\ast }}$,
there is $j^{\prime }$ who belongs to the same group of equals to $j$ and $%
\mu _{m}(j^{\prime })=c^{\prime }$ and $c_{m+1}P_{j^{\prime }}c^{\prime }$.
Moreover, since $j^{\prime }\sim _{c_{m+1}}j$ and $i_{m}^{\prime }\sim
_{c_{m+1}}i_{m},$ $j^{\prime }\succ _{c_{m+1}}i_{m}.$ These facts contradict
that $\left( i_{m},c_{m}\right) \blacktriangleright ^{\lambda ^{\mu ^{\ast
}}}\left( i_{m+1},c_{m+1}\right) $. Therefore, $i_{m}^{\prime }\succsim
_{c_{m+1}}j$ for all $(j,c^{\prime })$ satisfying $(j,c^{\prime })\gtrdot
^{\lambda ^{\mu ^{\ast }}}\left( i_{m+1}^{\prime },c_{m+1}\right) $.

Hence 
\begin{equation*}
\left( i_{m}^{\prime },c_{m}\right) \blacktriangleright ^{\lambda ^{\mu
^{\ast }}}\left( i_{m+1}^{\prime },c_{m+1}\right) ,
\end{equation*}%
for all $m=1,\cdots ,M,$ where $\left( i_{M+1}^{\prime },c_{M+1}\right)
=\left( i_{1}^{\prime },c_{1}\right) $. If $\left( i_{1}^{\prime
},c_{1}\right) ,$ $\left( i_{2}^{\prime },c_{2}\right) ,\ldots ,\left(
i_{M}^{\prime },c_{M}\right) $ are distinct, then they form an ex ante
stable improvement cycle at $\lambda ^{\mu ^{\ast }}$. If they are not
distinct, then the closed sequence still contains an ex ante stable
improvement cycle at $\lambda ^{\mu ^{\ast }}$. This contradicts Lemmas 1
and 4, which imply that there is no ex ante stable improvement cycle at $%
\lambda ^{\mu ^{\ast }}$. \textbf{Q.E.D.}\newline

In summary, the following simple method yields an ex ante stable assignment
that is not ordinally dominated by any other ex ante lottery and satisfies
equal treatment of equals. First, we derive a constrained stable matching $%
\mu ^{\ast }$, for example by using the stable improvement cycle mechanism
of Erdil and Ergin (2008) or the efficiency-adjusted deferred acceptance
mechanism of Kesten (2010). Second, we derive the ETE reassignment of $%
\lambda ^{\mu ^{\ast }}$.

\section{Ex post stability}

In this section, we consider a weaker notion of stability. A lottery $%
\lambda $ is said to be \textbf{ex post stable} if every matching in $%
\mathcal{M}\left( \lambda \right) $ is stable. If a lottery $\lambda $ is ex
ante stable, then it is also ex post stable.

We now reconsider the lottery used in Theorem 1. Let $\mu ^{\ast }$ be a
constrained efficient matching. We ask whether the ETE reassignment of $%
\lambda ^{\mu ^{\ast }}$ is not dominated by any \textit{ex post} stable
lottery.\footnote{%
Aziz et al. (2026) study a related problem: improving a random matching
while preserving ex post stability. Although their framework can accommodate
ETE, the resulting random matching need not be ex ante stable. Moreover,
their approach relies on computationally demanding optimization methods,
rather than providing a polynomial-time construction.} The following example
shows that this need not be the case: for some school-choice problem and
some constrained efficient matching $\mu ^{\ast }$, the ETE reassignment of $%
\lambda ^{\mu ^{\ast }}$ is dominated by an ex post stable lottery.\footnote{%
Example 1 is suggested by Minoru Kitahara, and I am grateful for his
contribution.}

\subsubsection*{Example 1}

Let the school choice problem be such that $I_{1}=\left\{ i,i^{\prime
}\right\} $, $I_{2}=\left\{ j,j^{\prime }\right\} $, $I_{3}=\left\{
k\right\} $, $I_{4}=\left\{ l\right\} ,$ $C=\left\{ a,b,c,d\right\} $, $%
q_{a}=q_{c}=2$, and $q_{b}=q_{d}=1$. The students' preferences are given in
the following table, where each column lists the schools in descending order
of preference.

\begin{center}
\begin{tabular}{llllll}
$P_{i}$ & $P_{i^{\prime }}$ & $P_{j}$ & $P_{j^{\prime }}$ & $P_{k}$ & $P_{l}$
\\ \hline
$a$ & $a$ & $a$ & $a$ & $b$ & $d$ \\ 
$b$ & $b$ & $c$ & $c$ & $d$ & $b$ \\ 
$c$ & $c$ & $b$ & $b$ & $a$ & $a$ \\ 
$d$ & $d$ & $d$ & $d$ & $c$ & $c$%
\end{tabular}
\end{center}

The priority orders of schools are as follows: 
\begin{gather*}
i\sim _{a}i^{\prime }\sim _{a}j\sim _{a}j^{\prime }, \\
l\succ _{b}i\sim _{b}i^{\prime }\succ _{b}k, \\
k\succ _{d}l.
\end{gather*}%
The remaining priority comparisons can be completed arbitrarily, provided
that they do not affect the comparisons specified above.

First, let $\mu ^{\ast }$ be such that 
\begin{equation*}
\mu ^{\ast }\left( i\right) =\mu ^{\ast }\left( j\right) =a,\text{ }\mu
^{\ast }\left( i^{\prime }\right) =\mu ^{\ast }\left( j^{\prime }\right) =c,%
\text{ }\mu ^{\ast }\left( k\right) =d,\text{ }\mu ^{\ast }\left( l\right) =b%
\text{.}
\end{equation*}%
It is straightforward to show that $\mu ^{\ast }$ is constrained efficient.
We derive the ETE reassignment of $\lambda ^{\mu ^{\ast }}$. Then, $\mu
^{1},\ldots ,\mu ^{4}$ that satisfy$\ \mu ^{1}=\mu ^{\ast }$, 
\begin{eqnarray*}
\mu ^{2}\left( i^{\prime }\right)  &=&\mu ^{2}\left( j\right) =\mu
^{3}\left( i^{\prime }\right) =\mu ^{3}\left( j^{\prime }\right) =\mu
^{4}\left( i\right) =\mu ^{4}\left( j^{\prime }\right) =a, \\
\mu ^{2}\left( i\right)  &=&\mu ^{2}\left( j^{\prime }\right) =\mu
^{3}\left( i\right) =\mu ^{3}\left( j\right) =\mu ^{4}\left( i^{\prime
}\right) =\mu ^{4}\left( j\right) =c, \\
\mu ^{n}\left( k\right)  &=&d\text{ and }\mu ^{n}\left( l\right) =b,\text{
for all }n=1,\ldots ,4,
\end{eqnarray*}%
are derived from $\mu ^{\ast }$. Therefore, the ETE reassignment of $\lambda
^{\mu ^{\ast }}$ denoted by $\lambda ^{\ast }$ satisfies 
\begin{equation*}
\lambda ^{\ast }\left( \mu ^{n}\right) =0.25\text{ for all }n=1,\ldots ,4.
\end{equation*}%
Thus,%
\begin{eqnarray*}
\Pr (a,\mathbf{c}\left( i,\lambda ^{\ast }\right) ) &=&\Pr (a,\mathbf{c}%
\left( i^{\prime },\lambda ^{\ast }\right) )=\Pr (c,\mathbf{c}\left(
i,\lambda ^{\ast }\right) )=\Pr (c,\mathbf{c}\left( i^{\prime },\lambda
^{\ast }\right) )=\frac{1}{2}, \\
\Pr (a,\mathbf{c}\left( j,\lambda ^{\ast }\right) ) &=&\Pr (a,\mathbf{c}%
\left( j^{\prime },\lambda ^{\ast }\right) )=\Pr (c,\mathbf{c}\left(
j,\lambda ^{\ast }\right) )=\Pr (c,\mathbf{c}\left( j^{\prime },\lambda
^{\ast }\right) )=\frac{1}{2}.
\end{eqnarray*}

On the other hand, let 
\begin{eqnarray*}
\bar{\mu}\left( i\right) &=&\bar{\mu}\left( i^{\prime }\right) =a,\text{ }%
\bar{\mu}\left( j\right) =\bar{\mu}\left( j^{\prime }\right) =c,\text{ }\bar{%
\mu}\left( k\right) =b,\text{ }\bar{\mu}\left( l\right) =d\text{,} \\
\bar{\mu}^{\prime }\left( i\right) &=&\bar{\mu}^{\prime }\left( i^{\prime
}\right) =c,\text{ }\bar{\mu}^{\prime }\left( j\right) =\bar{\mu}^{\prime
}\left( j^{\prime }\right) =a,\text{ }\bar{\mu}^{\prime }\left( k\right) =d,%
\text{ }\bar{\mu}^{\prime }\left( l\right) =b\text{.}
\end{eqnarray*}%
Then, $\bar{\mu}$ and $\bar{\mu}^{\prime }$ are stable. Let $\lambda ^{\ast
\ast }$ be such that $\lambda _{\bar{\mu}}^{\ast \ast }=\lambda _{\bar{\mu}%
}^{\ast \ast }=0.5$. Then, $\lambda ^{\ast \ast }$ has ex ante justified
envy, because $\bar{\mu}^{\prime }\left( i\right) =c$, $\bar{\mu}\left(
k\right) =b,$ $bP_{i}c$, and $i\succ _{b}k$. Therefore, $\lambda ^{\ast \ast
}$ is ex post stable but not ex ante stable.

Moreover, since $bP_{k}d$ and $dP_{l}b$, $\lambda ^{\ast }$ is ordinally
dominated by $\lambda ^{\ast \ast }$. Therefore, the ETE reassignment of a
lottery that assigns probability one to a constraint efficient matching is
dominated by some ex post stable lottery.

Note that the result of the FDAT introduced by Kesten and \"{U}nver (2015)
may also be dominated by some ex post stable lottery. In Example 1, the FDAT
results in a lottery $\lambda ^{\prime }$ satisfies 
\begin{eqnarray*}
\Pr (a,\mathbf{c}\left( i,\lambda ^{\prime }\right) ) &=&\Pr (a,\mathbf{c}%
\left( i^{\prime },\lambda ^{\prime }\right) )=\Pr (c,\mathbf{c}\left(
i,\lambda ^{\prime }\right) )=\Pr (c,\mathbf{c}\left( i^{\prime },\lambda
^{\prime }\right) )=\frac{1}{2}, \\
\Pr (a,\mathbf{c}\left( j,\lambda ^{\prime }\right) ) &=&\Pr (a,\mathbf{c}%
\left( j^{\prime },\lambda ^{\prime }\right) )=\Pr (c,\mathbf{c}\left(
j,\lambda ^{\prime }\right) )=\Pr (c,\mathbf{c}\left( j^{\prime },\lambda
^{\prime }\right) )=\frac{1}{2}, \\
\Pr (d,\mathbf{c}\left( k,\lambda ^{\prime }\right) ) &=&1,\text{ and }\Pr
(b,\mathbf{c}\left( l,\lambda ^{\prime }\right) )=1\text{.}
\end{eqnarray*}%
Therefore, $\lambda ^{\prime }$ is also ordinally dominated by $\lambda
^{\ast \ast }$.\footnote{%
Strictly speaking, FDAT is designed to obtain a random matching, that is, a
stochastic matrix, which can be induced by an ex ante stable lottery that is
not ordinally dominated by any other ex ante stable lottery. Since several
lotteries may induce the same random matching, the lottery $\lambda ^{\ast }$
considered above can be regarded as one such lottery.}

On the other hand, the lottery $\lambda ^{\bar{\mu}}=1$ satisfies ex ante
stability and ETE, and is not ordinally dominated by any ex post stable
lottery. We can generalize this result; that is, there always exists a
stable matching $\mu $ such that the ETE reassignment $\lambda $ with $%
\lambda _{\mu }=1$ is not ordinally dominated by any ex post stable lottery.
To identify such a stable matching, we use a rank-based criterion.

Let $r_{i}\left( c\right) =\left\vert \left\{ \left. d\in C\text{ }%
\right\vert \text{ }dP_{i}c\right\} \right\vert +1$ that represents the rank
position of $c$ in the preference ranking of student $i$. Moreover, let 
\begin{equation*}
R\left( \lambda \right) =\sum\limits_{i\in I}\sum\limits_{c\in C}\Pr (c,%
\mathbf{c}\left( i,\lambda \right) )\times r_{i}\left( c\right) ,
\end{equation*}%
which is the expected value of the sum of rank positions of the assignment
in the preference rankings of all agents.\footnote{%
Rank-based measures of this kind are also relevant in practice, since
outcomes in matching systems are sometimes evaluated using aggregate
statistics of assigned ranks. Accordingly, the market design literature has
used such measures as efficiency criteria; see, for example, Featherstone
(2020).}

In Example 1, we have

\begin{equation*}
R\left( \lambda ^{\bar{\mu}}\right) =8<11=R\left( \lambda ^{\mu ^{\ast
}}\right) .
\end{equation*}%
It is easy to verify that the ETE reassignment of a lottery does not change
the value of $R$. Indeed, if $\bar{\lambda}$ and $\lambda ^{\ast \ast }$ are
the ETE reassignments of $\lambda ^{\bar{\mu}}$ and $\lambda ^{\mu ^{\ast }},
$ respectively, then 
\begin{equation*}
R\left( \bar{\lambda}\right) =8<11=R\left( \lambda ^{\ast \ast }\right) 
\text{.}
\end{equation*}%
Therefore, if $\lambda ^{\bar{\mu}}$ minimizes $R$ among all deterministic
lotteries induced by stable matchings $\mu $, then there is no ex post
stable lottery $\lambda $ such that $R\left( \lambda \right) <R\left( \bar{%
\lambda}\right) $. These observations yield the following result.

\begin{proposition}
If $\bar{\mu}$ is a stable matching where $R\left( \lambda ^{\bar{\mu}%
}\right) \leq R\left( \lambda ^{\mu }\right) $ for any stable matching $\mu $%
, then the ETE reassignment of $\lambda ^{\bar{\mu}}$ is an ex ante stable
lottery that is not ordinally dominated by any other ex post stable lottery
and satisfies ETE.
\end{proposition}

\textbf{Proof.} We use the following result established by Okumura (2025).

\begin{lemma}
(Okumura 2025, Claim 1) If $\lambda $ is ordinally dominated by $\lambda
^{\prime },$ then $R\left( \lambda \right) >R\left( \lambda ^{\prime
}\right) $.
\end{lemma}

Let $\bar{\mu}$ be a stable matching where $R\left( \lambda ^{\bar{\mu}%
}\right) \leq R\left( \lambda ^{\mu }\right) $ for any stable matching $\mu $%
. Then, by Lemma 5, $\bar{\mu}$ is constrained efficient.

Next, let $\bar{\lambda}$ be the ETE reassignment of $\lambda ^{\bar{\mu}}$.
By Lemma 3 and Proposition 1, $\bar{\lambda}$ is an ex ante stable lottery
that satisfies ETE. We show that $\bar{\lambda}$ is not ordinally dominated
by any ex post stable lottery. Let $\lambda ^{\prime }$ be an arbitrary ex
post stable lottery.

First, by the construction of the ETE reassignment, $R\left( \lambda ^{\bar{%
\mu}}\right) =R\left( \bar{\lambda}\right) $. Next, 
\begin{equation*}
R\left( \lambda ^{\prime }\right) =\sum\limits_{\mu \in \mathcal{M}\left(
\lambda ^{\prime }\right) }\lambda _{\mu }^{\prime }\sum\limits_{i\in
I}r_{i}\left( \mu \left( i\right) \right) =\sum\limits_{\mu \in \mathcal{M}%
\left( \lambda ^{\prime }\right) }\lambda _{\mu }^{\prime }\times R\left(
\lambda ^{\mu }\right) .
\end{equation*}%
Since $\lambda ^{\prime }$ is an ex post stable lottery, any $\mu \in 
\mathcal{M}\left( \lambda ^{\prime }\right) $ is stable. Therefore, 
\begin{equation*}
R\left( \bar{\lambda}\right) =R\left( \lambda ^{\bar{\mu}}\right) \leq
R\left( \lambda ^{\prime }\right) \text{.}
\end{equation*}%
This result and Lemma 5 imply that $\bar{\lambda}$ is not ordinally
dominated by any ex post stable lottery. \textbf{Q.E.D.}\newline

Therefore, there exists a constrained efficient matching $\mu $ such that
the ETE reassignment of $\lambda ^{\mu }$ is ex ante stable and is not
ordinally dominated by any ex post stable lottery. However, this proof does
not provide a constructive procedure for identifying such a matching.
Developing a computationally efficient algorithm for finding such a matching
remains an open question.

\section*{References}

\begin{description}
\item Aziz, H., Biro, P., Cs\'{a}ji, G., Demeulemeester, T. 2026. Smart
Lotteries in School Choice: Ex-ante Pareto-Improvement with Ex-post
Stability. Unpublished manuscript available at arXiv:2602.10679.

\item Cookson, B., Shah, N. 2025. Fairly Stable Two-Sided Matching with
Indifferences. Proceedings of the 26th ACM Conference on Economics and
Computation, p.1130. https://doi.org/10.1145/3736252.3742675

\item Erdil, A., Ergin, H. 2008. What's the Matter with Tie-Breaking?
Improving Efficiency in School Choice. American Economic Review 98, 669--689.

\item Featherstone, C. 2020. Rank efficiency: Modeling a common policymaker
objective. Unpublished paper.

\item Kesten, O. 2010. School Choice with Consent. Quarterly Journal of
Economics 125, 1297--1348.

\item Kesten, O., Kurino, M., Nesterov, A.S. 2017. Efficient lottery design.
Social Choice and Welfare 48, 31--57.

\item Kesten, O., \"{U}nver, M. U. 2015. A Theory of School-Choice
Lotteries. Theoretical Economics 10, 543--595.

\item Okumura, Y. 2025. Equal Treatment of Equals and Efficiency in
Probabilistic Assignments. Unpublished manuscript available at arXiv
2508.14522.
\end{description}

\end{document}